\documentclass[11pt]{article}

\usepackage{amsmath,amsthm,amssymb,amsfonts,setspace}
\usepackage{bm}
\usepackage{bbold}
\usepackage{float}
\usepackage[colorlinks]{hyperref}
\usepackage{url}
\usepackage{fullpage}
\usepackage[affil-it]{authblk}

\newcommand{\eq}[1]{Eq.~\hyperref[eq:#1]{(\ref*{eq:#1})}}
\newcommand{\eqbeg}[1]{Equation~\hyperref[eq:#1]{(\ref*{eq:#1})}}
\newcommand{\eqs}[2]{Eqs.~\hyperref[eq:#1]{(\ref*{eq:#1})} and \hyperref[eq:#2]{(\ref*{eq:#2})}}
\renewcommand{\sec}[1]{\hyperref[sec:#1]{Section~\ref*{sec:#1}}}
\newcommand{\fig}[1]{\hyperref[fig:#1]{Fig.~\ref*{fig:#1}}}
\newcommand{\sfig}[2]{\hyperref[fig:#1]{Fig.~\ref*{fig:#1}#2}}
\newcommand{\thm}[1]{\hyperref[thm:#1]{Theorem~\ref*{thm:#1}}}
\newcommand{\lem}[1]{\hyperref[lem:#1]{Lemma~\ref*{lem:#1}}}

\newtheorem{theorem}{Theorem}

\newtheorem{lemma}{Lemma}

\begin{document}
\title{BosonSampling with Lost Photons}

\author[1]{Scott Aaronson\thanks{email: aaronson@csail.mit.edu. \ Supported by an
Alan T.\ Waterman Award from the National Science Foundation, under grant
no.\ 1249349.}}
\author[2]{Daniel J. Brod\thanks{email: dbrod@perimeterinstitute.ca. \ Research at Perimeter Institute
is supported by the Government of Canada through Industry Canada and by the Province of Ontario
through the Ministry of Research and Innovation.}}
\affil[1]{Massachusetts Institute of Technology, Cambridge, MA, USA}
\affil[2]{Perimeter Institute for Theoretical Physics, Waterloo, ON, Canada}

\date{}

\maketitle

\begin{abstract}
BosonSampling is an intermediate model of quantum computation where linear-optical networks are used to solve sampling problems expected to be hard for classical computers. \ Since these devices are not expected to be universal for quantum computation, it remains an open question of whether any error-correction techniques can be applied to them, and thus it is important to investigate how robust the model is under natural experimental imperfections, such as losses and imperfect control of parameters. \ 
Here we investigate the complexity of BosonSampling under photon losses---more specifically, the case where an unknown subset of the photons are randomly lost at the sources. \ We show that, if $k$ out of $n$ photons are lost, then we cannot sample classically from a distribution that is $1/n^{\Theta(k)}$-close (in total variation distance) to the ideal distribution, unless a $\text{BPP}^{\text{NP}}$ machine can estimate the permanents of Gaussian matrices in $n^{O(k)}$ time. \ In particular, if $k$ is constant, this implies that simulating lossy BosonSampling is hard for a classical computer, under exactly the same complexity assumption used for the original lossless case.
\end{abstract}

\section{Introduction} \label{sec:intro}

\textsc{BosonSampling} is a computational problem that involves sampling from a certain kind of probability distribution, defined in terms of the permanents of matrices. \ While it does not have any obvious applications, the \textsc{BosonSampling} problem has two features of interest: (i) it is ``easy'' to solve by a restricted, single-purpose quantum computer composed only of linear-optical elements, and (ii) it is hard to solve on a classical computer \cite{Aaronson2013a}, under plausible complexity-theoretic assumptions. \ Its natural implementation as a linear-optical experiment led, since the initial proposal, to a flurry of experiments of increasing complexity \cite{Broome2013, Crespi2013b, Tillmann2013, Spring2013, Spagnolo2013c, Carolan2014, Bentivegna2015, Carolan2015}.

The fact that these single-purpose linear-optical devices are not expected to be universal for quantum computing also means that they are not likely to have access to the usual tools of fault-tolerance and error correction, and thus it is important to investigate the resilience of the underlying models to real-world experimental imperfections. \ The original proposal of \cite{Aaronson2013a} already focused on this issue, giving evidence that \textsc{BosonSampling} is hard even if the classical computer is allowed to sample from a distribution that only approximates the ideal one (in variation distance). \ However, this does not suffice to address all realistic sources of error and imperfections that one may encounter. \ Thus, since then, further work has been done to investigate the robustness of \textsc{BosonSampling} devices to errors such as losses and mode mismatch \cite{Rohde2012,Shchesnovich2014}, imperfections on the linear-optical network \cite{Leverrier2015, Arkhipov2014}, and other sources of noise \cite{Kalai2014}.

There are a few reasons why photon losses are a central scalability issue in \textsc{BosonSampling} devices. \ The first is that they are pervasive in linear-optical experiments---to give an example, most recent \textsc{BosonSampling} experiments were performed using integrated photonic devices, where rather than setting up a linear-optical circuit using beam splitters in free space, the photons propagate along waveguides etched in a small chip. \ The most realistic models of loss in these devices give a per-photon probability of arrival that decreases exponentially with the depth of the circuit, and this is even before taking into account source and detector inefficiencies. \ The second reason is that often losses are dealt with using post-selection. \ By only accepting the experimental runs where the correct number of photons is observed, one can argue that the losses manifest only as an attenuation of the number of observed events per unit time. \ The problem with this approach, of course, is that even if each photon has only a constant probability of being lost, this already leads to an exponential overhead in experimental time as the desired number of photons increases.

In this paper, we give a formal treatment of the problem of \textsc{BosonSampling} with a few lost photons. \ More specifically, we show how the original argument of \cite{Aaronson2013a} can be modified to allow for the loss of a constant number (i.e.\ $O(1)$) of photons, thus showing that \textsc{BosonSampling} in this regime has the same complexity-theoretic status as in the ideal case. \ Of course, this falls short of the result one would want, since in practice one expects that at least a constant fraction (i.e.\ $\Theta(n)$) of the photons will be lost. \ Nevertheless, the result here constitutes a nontrivial step toward understanding experimentally-realistic \textsc{BosonSampling} regimes.

Let us point out in particular that, if more than $k=O(1)$ photons are lost, then we still get {\em some} hardness result: the issue is ``just'' that the strength of the result degrades rapidly as a function of $k$. \ More concretely, if $k$ out of $n$ photons are lost, then our argument will imply that \textsc{BosonSampling} cannot be done in classical polynomial time, to within a $1/n^{\Theta(k)}$ error in total variation distance, assuming that permanents of Gaussian matrices cannot be estimated in $n^{O(k)}$ time. \ So for example, even when a (sufficiently small) constant fraction of the photons are lost, $k=\epsilon n$, we can conclude that {\em exact} \textsc{BosonSampling} cannot be done in classical polynomial time, assuming that Gaussian permanent estimation is exponentially hard.

We also point out, in the appendix, how the proof of our main result can be repurposed to other physically-interesting error models, such as \textsc{BosonSampling} with dark counts, or a combination of both losses and dark counts.

\textbf{Notation}: Throughout this paper, we denote by $\mathcal{N}(0,1)_\mathbb{C}$ the complex Gaussian distribution with mean 0 and variance 1, and by $\mathcal{N}(0,1)^{m \times n}_\mathbb{C}$ the distribution over $m \times n$ matrices of i.i.d.\ complex Gaussian entries. \ For two probability distributions $\mathcal{D}_1 = \{p_x\}_x$ and $\mathcal{D}_2 = \{q_x\}_x$, the total variation distance $||\mathcal{D}_1 - \mathcal{D}_2||$ and the Kullback-Leibler divergence $D_{KL}(\mathcal{D}_{1}||\mathcal{D}_{2})$ are given, respectively, by:

\begin{align}
||\mathcal{D}_1 - \mathcal{D}_2|| := & \frac{1}{2} \sum_x |p_x - q_x|, \\
D_{KL}(\mathcal{D}_1 || \mathcal{D}_2) := & \sum_x p_x \ln \frac{p_x}{q_x}.
\end{align}

\section{Definition of the Problem} \label{sec:probdef}

Let $U$ be an $m \times m$ Haar-random unitary matrix (i.e.\ the ``interferometer''), and let $\Phi_{m,n}$ be the set of all lists of non-negative integers $S = (s_1, \ldots s_m)$ such that $\sum s_i = n$,  (i.e.\ the ``input/output states"). \ The transition probability from a given input state $S = (s_1, \ldots s_m) \in \Phi_{m,n}$ to an output state $T = (t_1, \ldots t_m) \in \Phi_{m,n}$ is given by
\begin{equation}\label{eq:permanent1}
\text{Pr}[S \rightarrow T] = \frac{|\text{Per}(U_{S,T})|^2}{s_1! \ldots s_m! t_1! \ldots t_m!},
\end{equation}
where $U_{S,T}$ is a submatrix of $U$ constructed as follows: (i) first, construct a $m \times n$ matrix $U_T$ by taking $t_1$ copies of the first row of $U$, $t_2$ copies of the second row, etc; (ii) then, form $U_{S,T}$ by taking $s_1$ times the first column of $U_T$, $s_2$ times its second column, etc. \ If we now fix some input state $S$ (which, for simplicity, we will take to be a string of $n$ ones followed by $m-n$ zeroes), \eq{permanent1} defines a distribution $\mathcal{D}_U$ given by
\begin{equation}\label{eq:permanent2}
\underset{\mathcal{D}_U}{\text{Pr}}[T] = \frac{|\text{Per}(U_{S,T})|^2}{t_1! \ldots t_m!}.
\end{equation}
The original \textsc{BosonSampling} problem then can simply be defined as producing a sample from $\mathcal{D}_U$, or at least from some other distribution that is close to it in total variation distance. \ It can also be shown that, if $m \gg n^2$ and $U$ is Haar-random, then one only needs to consider ``no-collision outputs,'' where each $t_i$ is only 0 or 1, since these dominate the distribution. \ For simplicity we will restrict our attention to these states from now on, replacing $\Phi_{m,n}$ by $\Lambda_{m,n}$ (defined as the subset of $\Phi_{m,n}$ with only no-collision states) and dropping the denominator of \eq{permanent2}, but for a full discussion see \cite{Aaronson2013a} or Sec.\ 4 of \cite{Aaronson2014}.

In this paper, we will consider a modified version of the \textsc{BosonSampling} problem, where a subset of the photons are lost along the way. \ Before giving the formal definition, consider the following illustrative example. \ Suppose we input one photon in each of the first $n+1$ modes of $U$ (i.e.\ $S \in \Lambda_{m,n+1}$ is a string of $n+1$ ones followed by all zeroes), but we observe a state with only $n$ photons at the output (i.e.\ $T \in \Lambda_{m,n}$). \ What probability should we ascribe to this event? Since all photons are identical and it is impossible to know which one was lost, the probability of this event is just the average of the probabilities of $n+1$ different \textsc{BosonSampling} experiments that use only $n$ out of the $n+1$ initial photons. \ In other words, we can write this as
\begin{equation}
\text{Pr}[T] = \frac{1}{n+1} \sum_i |\text{Per}(U_{S_i,T})|^2,
\end{equation}
where $U_{S_i,T}$ is obtained by deleting the $i$th column of $U_{S,T}$. \ Note that, by the way we defined $U_{S,T}$ it is, in this setting, an $n \times (n+1)$ matrix, so after deleting one of its columns we obtain a square matrix and the permanent function is well-defined.

We can now easily generalize the above to the case where exactly $k$ out of $n+k$ photons were lost. \ In this case, the probability of outcome $T \in \Lambda_{m,n}$ is given by
\begin{equation}
\text{Pr}[T] = \frac{1}{|\Lambda_{n+k,n}|} \sum_{S \in \bar{\Lambda}_{m,n+k,n}} |\text{Per}(U_{S,T})|^2,
\end{equation}
where now $\bar{\Lambda}_{m,k,n}$ denotes the set of all possible no-collision states of $n$ photons in the first $n+k$ out of $m$ modes. \ It is easy to see that $|\bar{\Lambda}_{m,n+k,n}| = |\Lambda_{n+k,n}| =\binom{n+k}{n}$.

Two remarks are in order about this particular choice of loss model. \ The first is that we are supposing that exactly $k$ photons have been lost in the experiment, but the more physically realistic model is one where each photon has an independent probability $p$ of being lost. \ These models may seem very different at first glance, especially since the latter is described by a state which does not even have a fixed photon number. \ Nevertheless, we can connect them as follows: if we start with $N$ total photons, and each has probability $p$ of being lost, then the average number of lost photons is $k = p N$. \ Furthermore, we need only repeat the experiment $O(\sqrt{N})$ times to have a high probability of observing an outcome with this exact value of $k$. \ This shows that the model considered here can be simulated (with only polynomial overhead) by the realistic one if we set $p = k/N$---and hence, that a hardness result for our model carries over to the realistic one, which is what we need.

The second remark is that, to simplify the analysis, we are assuming that the photons are lost only at the input to the interferometer (e.g.\ in the sources). \ In reality, however, we should expect them to be lost inside the circuit or at the detectors as well, in which case the above equation for the probability would be different. \ We will return to this issue in Section \ref{SUMMARY}.

Given our choice of loss model, we can view the modified problem more abstractly as \textsc{BosonSampling} with the $|\text{Per}(X)|^2$ function replaced by the following function:
\begin{equation}\label{eq:phi}
\Phi(A) := \frac{1}{|\Lambda_{n+k,n}|} \sum_S |\text{Per}(A_S)|^2,
\end{equation}
where $\Phi(A)$ is a function defined on $n \times (n+k)$ matrices, and the sum is taken over all square $n \times n$ proper submatrices of $A$.\footnote{Isaac Chuang has suggested to us that $\Phi(A)$ be called the ``temperament,'' perhaps because it generalizes the mod-squared permanent to more temperamental experimental apparatus.}

The main questions we address in this paper, then, are: is this generalized version of \textsc{BosonSampling}, where probabilities are given by objects like $\Phi(A)$, as hard to simulate approximately as the non-lossy case where probabilities are given by $|\text{Per}(X)|^2$? \ For what values of $k$?

We will show that the first question has an affirmative answer, but alas, we are only able to show a strong hardness result in the case where $k$ is a constant (i.e.\ does not grow with $n$). \ We leave, as our main open problem, to give a fuller understanding of what happens in the realistic case that $k$ grows with $n$.

Let us first give a (very) brief outline of the reasoning behind the result of Aaronson and Arkhipov \cite{Aaronson2013a}. We begin by defining the following problem:

\textbf{Problem 1} ($|$GPE$|^2_\pm$) \cite{Aaronson2013a}. Given as input a $n \times n$ matrix $X \sim \mathcal{N}(0,1)^{n \times n}_\mathbb{C}$ of i.i.d.\ Gaussians, together with error bounds $\epsilon, \delta >0$, estimate $|$Per$(X)|^2$ to within additive error $\pm \epsilon \cdot n!$ with probability at least $1-\delta$ over $X$ in poly($n,1/\epsilon,1/\delta$) time.

$|$GPE$|^2_\pm$ is the main problem addressed in \cite{Aaronson2013a}. \ Aaronson and Arkhipov show that, if there is a classical algorithm that efficiently produces a sample from a distribution close in total variation distance to that sampled by a \textsc{BosonSampling} device, then $|$GPE$|^2_\pm \in$ BPP$^{\text{NP}}$. \ They then conjecture that $|$GPE$|^2_\pm$ is $\#$P-hard (by means of two natural conjectures, the Permanent-of-Gaussians Conjecture and the Permanent Anti-Concentration Conjecture; see \cite{Aaronson2013a} for more details). \ Under that conjecture, they conclude that, if there was an efficient classical algorithm $C$ as described above, the polynomial hierarchy would collapse. \ This argument provides evidence that, despite being a very restricted quantum computer, a \textsc{BosonSampling} device is exponentially hard to simulate on a classical computer.

Here, we are going to ``intercept'' the result of Aaronson and Arkhipov halfway through: we will show that $|$GPE$|^2_\pm$ can be reduced to another problem, which is its natural generalization when we replace $|\text{Per}(X)|^2$ by $\Phi(A)$, as long as $k$ is constant. \ More concretely, consider the following problem:

\textbf{Problem 2} ($\Sigma |$GPE$|^2_\pm$). Given as input a $n \times (n+k)$ matrix $A  \sim \mathcal{N}(0,1)^{n \times (n+k)}_\mathbb{C}$, together with error bounds $\epsilon', \delta > 0$, estimate $\Phi(A) := \Sigma_{S} |\text{Per}(A_S)|^2$ to within additive error $\pm \epsilon' n!$ with probability at least $1-\delta$ over $A$ in poly($n,1/\epsilon',1/\delta$) time.

Then our main result is:

\begin{theorem} \label{thm:reduction}
If $\mathcal{O}$ is an oracle that solves $\Sigma |\text{\emph{GPE}}|^2_\pm$ with $\epsilon' = O\left(\frac{\epsilon \delta^{k+1/2} k^{k/2}}{n^{k/2}(n+k)^{k}}\right)$, then $|\text{\emph{GPE}}|^2_\pm$ can be solved in \emph{BPP}$^\mathcal{O}$.
\end{theorem}

Notice that the precision demanded of the oracle $\mathcal{O}$ is exponential in $k$, which is why we can only use \thm{reduction} to make claims about \textsc{BosonSampling} if $k$ is constant. \ However, in this case \thm{reduction} immediately allows us to replace all further claims made by Aaronson and Arkhipov regarding $|$GPE$|^2_\pm$ by their equivalent versions with $\Sigma |$GPE$|^2_\pm$, and thus show that \textsc{BosonSampling} with a few lost photons has the same complexity as the original \textsc{BosonSampling} problem.

\section{Reduction from \texorpdfstring{$|$GPE$|^2_\pm$}{} to \texorpdfstring{$\Sigma |$GPE$|^2_\pm$}{}} \label{sec:reductheorem}

In this section, we prove \thm{reduction}.

The idea behind the proof is as follows. \ Given $X \sim \mathcal{N}(0,1)^{n \times n}_\mathbb{C}$, we can trivially embed $X$ as the leftmost submatrix of another matrix $A  \sim \mathcal{N}(0,1)^{n \times (n+k)}_\mathbb{C}$. \ Now define $A[c]$ to be the matrix obtained by multiplying the $k$ rightmost columns of $A$ by the real number $c$, and let $\mathcal{N}_{A}[c]$ be the resulting distribution over $A[c]$. \ Then it is easy to see that

\begin{equation} \label{eq:phidef}
\Phi(A[c]) := \frac{1}{|\Lambda_{n+k,n}|} |\text{Per}(X)|^2 + |c|^2 Q_1 + |c|^4 Q_2 + .... + |c|^{2k} Q_k,
\end{equation}
where each $Q_i$ is the sum of the absolute squares of the permanents of all submatrices of $A$ that include $i$ out of the $k$ rightmost columns of $A$. \ If $c$ is sufficiently close to $1$, then $\mathcal{N}_{A}[c]$ is close to $\mathcal{N}_{A}[1]$ \big(which coincides with $\mathcal{N}(0,1)^{n \times (n+k)}_\mathbb{C}$\big), and we can use the oracle $\mathcal{O}$ to estimate $\Phi(A[c])$. \ By calling $\mathcal{O}$ for $k+1$ different values of $c$, we can then estimate $|\text{Per}(X)|^2$ using standard polynomial interpolation.

Let us start with the following simple lemma, which relates the total variation distance between $\mathcal{N}_{A}[c]$ and $\mathcal{N}_{A}[1]$ to the distance between $c$ and $1$.

\begin{lemma} \label{lem:cbound}
If $|c-1| \leq \frac{\delta}{\sqrt{n k}}$, then $||\mathcal{N}_{A}[c]-\mathcal{N}_{A}[1]|| = \text{O}(\delta)$.
\end{lemma}
\begin{proof}
First notice that, since $\mathcal{N}_{A}[1]$ is just the joint distribution of $n(n+k)$ i.i.d.\ complex Gaussians $x_i \sim \mathcal{N}(0,1)_\mathbb{C}$, then $\mathcal{N}_{A}[c]$ is the joint distribution over $n(n+k)$ i.i.d.\ complex Gaussians $y_i \sim \mathcal{N}(0,\sigma_i)_\mathbb{C}$, where $\sigma_i$ is 1 for $n^2$ of the variables (those corresponding to the leftmost $n \times n$ submatrix of $A$) and $c$ for the remaining $nk$ variables. \ By Pinsker's inequality, we have
\begin{equation*}
||\mathcal{N}_{A}[c]-\mathcal{N}_{A}[1]|| \leq \sqrt{\frac{1}{2} D_{KL}(\mathcal{N}_{A}[1]||\mathcal{N}_{A}[c])},
\end{equation*}
where $D_{KL}(P||Q)$ is the Kullback-Leibler divergence. \ When $P$ and $Q$ are multivariate Gaussian distributions, there is a closed form for $D_{KL}(P||Q)$---in particular, if $P$ (respectively $Q$) corresponds to the distribution over $K$ i.i.d.\ complex variables $\{x_i\}$ with means 0 and corresponding variances $\{\sigma_{P,i}\}$ (respectively $\{\sigma_{Q,i}\}$), we can write \cite{Kullbackbook}:
\begin{equation*}
D_{KL}(P||Q) = \sum_i{\left(\frac{\sigma_{P,i}}{\sigma_{Q,i}}\right)^2} - K + 2 \sum_i \ln{\frac{\sigma_{Q,i}}{\sigma_{P,i}}}.
\end{equation*}

By setting $K$ equal to $n(n+k)$, $\sigma_{Q,i}$ equal to 1 for all $i$, and $\sigma_{P,i}$ equal to 1 for $n^2$ of the $i$'s and $c$ for the other $kn$ we get
\begin{equation*}
D_{KL}(\mathcal{N}_{A}[1]||\mathcal{N}_{A}[c])= n k \left(\frac{1}{c^2}-1+2 \ln c\right).
\end{equation*}
If $|1-c|=a$, we get
\begin{equation*}
D_{KL}(\mathcal{N}_{A}[1]||\mathcal{N}_{A}[c]) = 2 n k a^2 + O(a^3)
\end{equation*}
and hence
\begin{equation*}
||\mathcal{N}_{A}[c]-\mathcal{N}_{A}[1]|| \leq a \sqrt{n k}+O(a^2).
\end{equation*}
Setting $a = \frac{\delta}{\sqrt{n k}}$ completes the proof.
\end{proof}

By \lem{cbound}, to satisfy the definition of $\mathcal{O}$ we must have $|1-c| = O(\delta/\sqrt{kn})$. \ One might worry about how stable the estimate of $|\text{Per}(X)|^2$ produced by the polynomial interpolation will be, if we are only allowed to probe a very small region of values of $c$. \ We now show how to relate the precision in the output of $\mathcal{O}$ to the resulting estimate of $|\text{Per}(X)|^2$.

To begin, consider the equivalent problem of estimating the parameters $\{\beta_i\}$ of the polynomial
\begin{equation}
w = \beta_0 + \beta_1 x + \beta_2 x^2 + \ldots \beta_k x^k
\end{equation}
by choosing $k+1$ distinct values $x_i$ in the interval $[1-a,1+a]$ and estimating the corresponding values $w_i$ each with error $e_i < \epsilon' n!$. \ This results in the following linear system (written in vector notation):
\begin{equation}
\mathbf{w} = X \mathbf{\beta} + \mathbf{e},
\end{equation}
where $X$ is the Vandermonde matrix
\begin{equation} \label{NonlocalCore}
X =
\left(\begin{array}{ccccc}
1 & x_1 & x_1^2 & \ldots & x_1^k \\
1 & x_2 & x_2^2 & \ldots & x_2^k \\
1 & x_3 & x_3^2 & \ldots & x_3^k \\
\vdots & \vdots & \vdots & \ddots & \vdots \\
1 & x_{k+1} & x_{k+1}^2 & \ldots & x_{k+1}^k \\
\end{array}\right).
\end{equation}
Note that $X$ is invertible as long as we choose all $x_i$'s distinct.

We can now use, as an estimator for $\mathbf{\beta}$, the one given by the ordinary least squares method, namely
\begin{equation} \label{eq:leastsquares}
\mathbf{\hat{\beta}} = \left ( X^T X \right ) ^{-1} X^T \mathbf{w}.
\end{equation}
Let us now bound the variance on the estimator $\hat{\beta}_0$, which corresponds to the parameter that we are trying to estimate (i.e.\ $|\text{Per}(X)|^2$ in \eq{phidef}). \ Since the errors $e_i$ are all in the interval $[ -\epsilon' n!,\epsilon' n!]$ by assumption we can, without loss of generality, assume that  $\mathrm{E}(e_i) = 0$ and Var$(e_i) \leq (\epsilon' n!)^2$. \ Let us also denote by $\Omega$ the covariance matrix of the error vector $\mathbf{e}$. \ Then we can write:
\begin{equation}
\textrm{Var}  (\mathbf{\hat{\beta}})  = \textrm{Var} \left [ \left (X^T X \right )^{-1} X^T \mathbf{w} \right ] = X^{-1} \Omega \left ( X^{-1} \right )^{T},
\end{equation}
and thus
\begin{equation} \label{eq:varbeta1}
\textrm{Var}(\hat \beta_0) = \left [X^{-1} \Omega \left (X^{-1} \right )^{T} \right ]_{1,1}.
\end{equation}

But now we can bound this as:
\begin{align}
\left [ X^{-1} \Omega \left  ( X^{-1} \right )^{T} \right ]_{1,1} & \leq \max_{i,j} \left [ X^{-1} \Omega \left ( X^{-1} \right )^{T} \right ]_{i,j} \notag \\
& \leq ||X^{-1} \Omega (X^{-1})^{T}||_{\infty} \notag \\
& \leq ||X^{-1}||_{\infty} ||\Omega||_{\infty} ||(X^{-1})^T||_{\infty} \notag \\
& = ||X^{-1}||_{\infty} ||\Omega||_{\infty} ||(X^{-1})||_{1} \notag \\
& \leq (k+1) ||X^{-1}||_{\infty}^2 ||\Omega||_{\infty}, \label{eq:normbound1}
\end{align}
where $||A||_{\infty}$ and $||A||_1$ are the maximum row 1-norm and maximum column 1-norm, respectively, and we used the inequalities $||A B||_{\infty} \leq ||A ||_{\infty} || B ||_{\infty}$ and $||A||_1 \leq (k+1) ||A ||_{\infty}$, which hold for all $\{A,B\} \in M_{k+1}$ \cite{matrixanalysis}. \ We can now use a result due to Gautschi \cite{Gautschi1962}, which bounds the norm of the inverse of the square Vandermonde matrix as
\begin{equation*}
||X^{-1}||_{\infty} \leq \max_{1 \leq j \leq k+1} \prod_{i=1, i \neq j}^{k+1} \frac{1+|x_j|}{|x_j - x_i|}.
\end{equation*}

Since, in our case, all $x_i's$ are bounded in the interval $[1-a,1+a]$, we can write
\begin{equation} \label{eq:boundVander}
||X^{-1}||_{\infty} \leq (2+a)^k \max_{1 \leq j \leq k+1} \prod_{i=1, i \neq j}^{k+1} \frac{1}{|x_j - x_i|}.
\end{equation}
In order to obtain the sufficiently tight bound from the expression above, it is helpful to choose the $x_i$'s evenly spaced in the interval $[1-a,1+a]$. \ In that case, the maximum in \eq{boundVander} is obtained by choosing $x_j$ to be one of the central points, i.e.\ $j = k/2+1$ if $k$ is even or $j = (k+1)/2$ is $k$ is odd. \ In this case, it is not hard to show that
\begin{equation} \label{eq:normbound2}
||X^{-1}||_{\infty} \leq
  \begin{cases}
      \hfill \left(\frac{2 e}{a}\right)^k \frac{1}{\pi k}   \hfill & \text{ if $k$ is even} \\
      \hfill \left(\frac{2 e}{a}\right)^k \frac{1}{\pi \sqrt{k^2-1}}\frac{k^k}{\sqrt{(k-1)^{k-1}(k+1)^{k+1}}} \hfill & \text{ if $k$ is odd.} \\
  \end{cases}
\end{equation}
Clearly if $k$ is sufficiently large, these bounds coincide. \ Finally note that, since every $e_i$ has variance at most $(\epsilon' n!)^2$, we can write
\begin{equation} \label{eq:normbound3}
||\Omega||_{\infty} \leq (k+1) (\epsilon' n!)^2,
\end{equation}
By plugging Eqs.\ (\ref{eq:normbound1}), (\ref{eq:normbound2}) and (\ref{eq:normbound3}) back into \eq{varbeta1} we obtain
\begin{equation} \label{eq:varbeta2}
\textrm{Var}(\hat \beta_0) = O \left ( \frac{(\epsilon' n!)^2}{a^{2k}} \right ).
\end{equation}

We are now in position to prove \thm{reduction}.
\begin{proof}
We begin by writing
\begin{equation*}
\Phi(A[c]) := \frac{1}{|\Lambda_{n+k,n}|} |\text{Per}(X)|^2 + |c|^2 Q_1 + |c|^4 Q_2 + .... + |c|^{2k} Q_k.
\end{equation*}
We then choose $k+1$ values of $c$ equally spaced in the interval $\left [1-\delta/\sqrt{n k},1+\delta/\sqrt{n k} \right ]$. \ By \lem{cbound}, each $A[c]$ obtained in this way is within $\delta$ of total variation distance to $\mathcal{N}(0,1)^{n \times (n+k)}_\mathbb{C}$, so we can use the oracle $\mathcal{O}$ to obtain an estimate of $\Phi(A[c])$ to within additive error $\pm \epsilon' n!$. \ Using the ordinary least squares method, we can then give an estimate $\hat P$ for $\frac{1}{|\Lambda_{n+k,n}|}|\text{Per}(X)|^2$ with variance given by
\begin{equation*}
\textrm{Var}(\hat P) = O\left(\frac{(\epsilon' n!)^2}{(\delta^2/{n k})^{k}}\right),
\end{equation*}
which is just \eq{varbeta2} where we set $a=O(\delta/\sqrt{n k})$. \ Finally, by Chebyshev's inequality, we have
\begin{equation*}
\textrm{Pr}\left(|\hat P - \frac{1}{|\Lambda_{n+k,n}|}|\text{Per}(X)|^2| \geq K \sqrt{\textrm{Var}(\hat P)}\right) \leq 1/K^2.
\end{equation*}
For any $K>0$. \ Thus, by setting $K = 1/\sqrt{\delta}$ and $\epsilon' = O\left(\frac{\epsilon \delta^{k+1/2} k^{k/2}}{n^{k/2}(n+k)^{k}}\right)$, we obtain an estimate for $|$Per$(X)|^2$ to within additive error $\pm \epsilon \cdot n!$ with probability at least $1-\delta$ (over $X$ and over the errors in the outputs of $\mathcal{O}$) in poly($n,1/\epsilon,1/\delta$) time, thus solving $|$GPE$|^2_\pm$.
\end{proof}

\thm{reduction} then guarantees, for example, that if $k$ is constant, then as long as the error in each estimate of $\mathcal{O}$ is $\ll 1/n^{3k/2}$, solving $\Sigma |$GPE$|^2_\pm$ is as least as hard as solving $|$GPE$|^2_\pm$.

\section{Summary and Open Problems}
\label{SUMMARY}

We investigated the loss tolerance of \textsc{BosonSampling} in the regime where only a few photons are lost. \ In this case, the output distribution can become very far (in total variation distance) from the ideal one, and so the original argument of Aaronson and Arkhipov \cite{Aaronson2013a} does not work. \ Nevertheless, we showed how the problem of estimating the permanent, which usually describes the outcome probabilities in the ideal \textsc{BosonSampling} model, can be reduced to the problem of estimating the quantity $\Phi$, defined in \eq{phi}, which describes the probabilities in the lossy \textsc{BosonSampling} model. \ For the regime where the number of lost photons is constant (or, alternatively, where each photon is lost with probability $\sim c/n$), this allows us to replace $|\text{Per}(X)|^2$ by $\Phi(A)$ in Aaronson and Arkhipov's argument, and thus to obtain similar complexity results as those in \cite{Aaronson2013a}. \ Note that it was not necessary to modify either the Permanent-of-Gaussians Conjecture or the Permanent Anti-Concentration Conjecture that Aaronson and Arkhipov \cite{Aaronson2013a} used.

Note that, if we restrict our attention to the loss model where each photon has an independent probability $p$ of being lost (cf.\ the discussion after \sec{probdef}), then there is a much simpler hardness argument that works for $k=O(1)$. \ Namely, in this regime the probability that {\em zero} photons will be lost is only polynomially small, so one could simply repeat the experiment a polynomial number of times, and postselect on observing all of the photons. \ As a related observation, if the scattering matrix $U$ can be chosen arbitrarily (rather than from the Haar measure), then even if exactly $k$ photons will be lost in each run, we could simply pick $U$ to act as the identity on $k$ of the $n$ modes that initially have $1$ photon each. \ If no photons are observed in those $k$ modes---which occurs with probability $\Theta(1/n^k)$---then we know that lossless \textsc{BosonSampling} has been performed on the remaining $n-k$ photons.

In our view, these observations underscore that our hardness result is only a first step, and that ultimately one wants a strong hardness result that holds even when $k$ is non-constant. \ On the other hand, it is important to understand that our result is not itself obtainable by any simple postselection trick like those described above. \ Rather, our result requires analyzing output probabilities that are sums of absolute squares of many permanents, and showing that these sums inherit some of the hardness of the permanent itself. \ By demonstrating that such reductions are possible, we hope our result will inspire further work on the complexity of lossy \textsc{BosonSampling}. \ We also point out that the reasoning behind our proof of \thm{reduction} can be used to deal with other issues besides just photon losses, such as e.g.\ the case of photon ``shuffling'' described in the appendix, where photon losses compounded with dark counts acts as noise in the probability distribution that cannot be dealt with just by postselection.

Our work leaves a few open questions. \ The obvious one is to investigate what happens in more realistic loss regimes, such as e.g.\ if a constant fraction of the photons is lost (i.e.\ in our notation, $k = \epsilon n$). In that case, \thm{reduction} requires the oracle $\mathcal{O}$ to produce an estimate of $\Phi$ to within error $1/n^{\Theta(\epsilon n)}$, so it does not allow us to make any strong complexity claims. \ It would be reasonable to expect that, if $k$ is too large, \textsc{BosonSampling} becomes classically simulable---in fact, this is easily seen to be true if all but $O(\log n)$ photons get lost, and it would be interesting if we could determine where the transition happens.

In that direction, an even more basic question that one can ask is whether the ``temperament,'' $\Phi(A)$, of an i.i.d.\ Gaussian matrix $A \sim \mathcal{N}(0,1)^{n \times (n+k)}_\mathbb{C}$, becomes concentrated around some simple classical estimator when $k$ is large compared to $n$. \ For example, one could look at $R(A)$, the product of the squared $2$-norms of the rows of $A$. \ Aaronson and Arkhipov \cite{Aaronson2014} previously studied this estimator in the case of {\em no} lost photons (i.e., $k=0$), and showed that even there, $R(A)$ has a constant amount of correlation with $|\text{Per}(A)|^2$. \ If (say) $n=2$, then it is not hard to see that $R(A)$ has an excellent correlation with $\Phi(A)$, as $k$ gets arbitrarily large. \ What happens when $n\sim \log^{2} k$ or $n\sim \sqrt{k}$?

A second open problem concerns our choice to model photon losses assuming all of them were lost at the input, as described at the beginning of \sec{probdef}. \ We leave as a question for future work to adapt our proof to the case where losses happen at the detectors. \ In that case, rather than summing over all possibilities that $k$ out of $k+n$ input photons were lost, we would have to sum over all possibilities of a subset of $k$ out of the $m-n$ detectors malfunctioning. \ When we do that, we obtain a probability expression that is very similar to that in \eq{phi}, but with three complications: (i) the permanents are of $(n+k) \times (n+k)$ matrices, since now $n+k$ photons actually traveled through the network, (ii) the sum is over $\binom{m-n}{k}$ possibilities, rather than $\binom{n+k}{n}$, and (iii) rather than considering $A$ to be a $n \times (n+k)$ matrix of i.i.d.\ Gaussian entries, we need to consider it an $m \times (n+k)$ random matrix of orthonormal columns. \ Nevertheless, in the regime $m = \Theta(n^2)$ that is often assumed in \textsc{BosonSampling}, it should be possible (although somewhat cumbersome) to adapt our polynomial interpolation argument to this case. \ Then, as a next step, one could consider the case where the photons can be randomly lost both at the sources {\em and} at the detectors, or indeed anywhere in the network.

\section*{Acknowledgments}
We thank Aram Harrow and Anthony Leverrier for helpful discussions.

\bibliographystyle{ieeetr}

\appendix
\section{Appendix: Other Applications of Main Result}

In this appendix, we point out how our main result in \sec{reductheorem} can be repurposed to include other scenarios besides photon losses, such as dark counts. \ In the same way that our main result (\thm{reduction}) only extends the loss-tolerance of \textsc{BosonSampling} to allow for a constant number $k$ of photons to be lost, the other scenarios described in this Appendix also have some parameter $k$ (e.g.\ the number of dark counts) that can be at most a constant.

Recall that our main result concerns the setting where $n+k$ photons were prepared, but only $n$ of them were measured, and the corresponding probability is given by replacing the permanent with $\Phi(A)$ defined in \eq{phi}, which is a function of $n \times (n+k)$ matrices. \ We then embed the $n \times n$ matrix $X$ whose permanent we wish to estimate in a matrix $A[c]$ defined as

\begin{equation}
A[c] = \left ( X \; c Y \right ),
\end{equation}
for some $k \times k$ Gaussian $Y$. \ We then write, as in \eq{phidef},
\begin{equation}
\Phi(A[c]) := \frac{1}{|\Lambda|} (|\text{Per}(X)|^2 + |c|^2 Q_1 + |c|^4 Q_2 + .... + |c|^{2k} Q_k),
\end{equation}
where each $Q_i$ is given in terms of permanents of different submatrices of $A[c]$, other than $X$, and their values are irrelevant for our purposes. \ Finally, we can use a polynomial regression to produce an estimate for $|\text{Per}(X)|^2$ from estimates of $\Phi(A[c])$ for several values of $c$ sufficiently close to 1.

Now, let us show how other scenarios can be cast in the same form.

\subsection{Dark Counts}

Suppose that, rather than preparing $n+k$ photons and measuring $n$ of them, we prepare $n$ and measure $n+k$ photons---that is, we have events known as dark counts, when an optical detector ``clicks'' even in the absence of a photon. \ This means that we must now ascribe, to our detected event, a probability that is the average of the probabilities obtained by considering that each possible subset $k$ of the $k+n$ detectors were the ones that misfired. \ In other words, we can write the probability as

\begin{equation}
\Phi_{\text{dark}}(A) :=\frac{1}{|\Lambda|} \Sigma_{T \in \Lambda} |\text{Per}(A_T)|^2
\end{equation}
where again $\Lambda$ denotes the set of all possible no-collision states of $n$ photon in $n+k$ modes, but $A$ is now a $(n+k) \times n$ matrix, and $A_T$ is the submatrix of $A$ obtained by taking $n$ of its rows corresponding to the 1's in $T$. \ We once more embed the $n \times n$ matrix $X$ in a larger matrix, but which now is defined as
\begin{equation}
A[c] = \left( \begin{array}{c}
X \\ cY
 \end{array} \right).
\end{equation}
Note that this leads to an identical expression as \eq{phidef},
\begin{equation}
\Phi_{\text{dark}}(A[c]) := =\frac{1}{|\Lambda|}(|\text{Per}(X)|^2 + |c|^2 Q_1 + |c|^4 Q_2 + .... + |c|^{2k} Q_k).
\end{equation}
As before, the $Q_i$'s correspond to permanents of different submatrices of $A[c]$, which are now distributed vertically rather than horizontally. \ It is clear from the form of $\Phi_{\text{dark}}(A[c])$ that the arguments of \sec{reductheorem} follow through unchanged.

\subsection{Combination of Losses and Dark Counts}

Consider now the following scenario: we prepared $n+k$ photons and measured $n+k$ photons, but during the process $k$ of our photons were lost, and there were $k$ dark counts (let us call this event of losing a photon and getting a dark count as ``shuffling'' a photon). \ One difference compared to the previous scenarios is that a shuffling event is not {\em heralded} in any way---in principle, any experiment which performs $n$-photon \textsc{BosonSampling} should take all probabilities of up to $n$-fold shuffling into account as noise in the final distribution. \ As we now argue, in a situation where up to $k$ out of $n+k$ photons can be shuffled, for constant $k$, the reasoning behind our hardness result goes through.

First, suppose that, for whatever reason, we know that exactly $k$ photons were shuffled during the process. \ In this case, the probability we ascribe to the observed event is given by
\begin{equation}
\Phi_{\text{shuf}}(A,k) :=\frac{1}{|\Lambda|^2} \Sigma_{S, T \in \Lambda} |\text{Per}(A_{ST})|^2,
\end{equation}
where now $S$ and $T$ both run over all possible collision-free basis states corresponding to $n$ photons in $k+n$ modes and index rows and columns, respectively, of $A$. \ In other words, again the sum is over all $n \times n$ proper submatrices of $A$, which now is $(n+k) \times (n+k)$. \ As before, we can build a matrix of the form
\begin{equation}
A[c] = \left( \begin{array}{cc}
X & c Y\\
c V & c^2 W \end{array} \right),
\end{equation}
and write
\begin{equation} \label{eq:shufflek}
\Phi_{\text{shuf}}(A[c],k)=\frac{1}{|\Lambda|^2} (|\text{Per}(X)|^2 + |c|^2 Q_1 + |c|^4 Q_2 + .... + |c|^{4k} Q_2k).
\end{equation}
In this case, each $Q_j$ is a sum that collects all submatrices of $A[c]$ which use $j$ rows and/or columns of $A$ that contain the factor $c$, and the total polynomial is of degree $4k$, rather than $2k$ as in previous cases. \ Nonetheless, it is straightforward to adapt the least squares method to this setting, and our main argument goes through with no significant changes.

However, since shufflings are not heralded, restricting ourselves to the case where exactly $k$ photons are shuffled is not very natural. \ To make it (slightly) more natural, we can consider the case where all $j$-fold shufflings can happen, from $j$ ranging from 1 to $k$, with corresponding known probabilities $p_j$. \ In this case, the probability we assign to the event is
\begin{equation} \label{eq:shuffle}
\Phi_{\text{shuf}}(A[c])=\sum_{j=0}^{k} p_j \Phi_{\text{shuf}}(A[c],j).
\end{equation}
But note that, in each $\Phi_{\text{shuf}}(A[c],j)$ for $j$ between 0 and $k-1$, the permanents in \eq{shufflek} are over submatrices of dimensions $(n+k-j) \times (n+k-j)$, and thus each of them must include at least one column or row that contains a factor $c$. \ This means that, in the sum in \eq{shuffle}, the only term that does not have a factor of $c$ is the one containing $|\text{Per}(X)|^2$. \ Thus we can write once again
\begin{equation}
\Phi_{\text{shuf}}(A[c],k)=\frac{1}{|\Lambda|^2} (p_k |\text{Per}(X)|^2 + |c|^2 Q'_1 + |c|^4 Q'_2 + .... + |c|^{4k} Q'_2k).
\end{equation}
where now all $Q'_i$s may depend in complicated but unimportant ways on all $j$-fold shuffling events. \ We can again apply the least-squares method, but now there is an important caveat: rather than obtaining an estimate of $|\text{Per}(X)|^2$ to within additive precision, we only obtain an estimate of $p_k |\text{Per}(X)|^2$. \ This introduces the additional restriction that $p_k$ is large enough that the additive approximation does not become useless.
	
\end{document}